\documentclass[11pt,a4paper,longbibliography]{amsart}
\usepackage[margin=2.5cm]{geometry}
\usepackage{cite}
\usepackage{amsmath,amsfonts,amssymb,amsthm,amscd,mathrsfs,mathdots,bbm,color,bm,enumerate,comment,graphicx,newtxmath,newtxtext,subcaption,stmaryrd,caption,soul}
\usepackage{comment}

\numberwithin{equation}{section}
\definecolor{equationcolor}{RGB}{222,94,100}
\definecolor{alecolor}{RGB}{198,113,190}

\usepackage[linecolor=white,backgroundcolor=orange,bordercolor=red]{todonotes}

\usepackage[colorlinks=true,linkcolor=equationcolor,citecolor=teal]{hyperref}

\hypersetup{urlcolor=teal}
\definecolor{equationcolor}{RGB}{222,94,100}

\DeclareFontFamily{U}{mathx}{}
\DeclareFontShape{U}{mathx}{m}{n}{<-> mathx10}{}
\DeclareSymbolFont{mathx}{U}{mathx}{m}{n}
\DeclareMathAccent{\widecheck}{0}{mathx}{"71}

\def\R{{\mathbb R}}

\def\N{{\mathbb N}}

\def\UT{\mathrm{UT}}

\def\AA{\mathcal{A}}

\def\SS{\mathcal{S}}

\def\de{ \mathrm{d}}

\def\conv{\operatorname{conv}}

\usepackage{thmtools,amsthm}
\theoremstyle{plain}

\newtheorem{thm}{Theorem}[section]

\newtheorem{cor}[thm]{Corollary}
\newtheorem{prop}[thm]{Proposition}
\newtheorem{lem}[thm]{Lemma}

\newtheorem{ex}{Example}[section]

\newtheorem{mainthm}{Theorem}

\theoremstyle{definition}

\theoremstyle{remark}
\newtheorem{rem}{Remark}

\newcommand{\be}{\begin{equation}}
\newcommand{\ee}{\end{equation}}

\newcommand{\ben}{\begin{equation*}}
\newcommand{\een}{\end{equation*}}

\def\_#1{\def\next{#1}
 \ifx\next\risingsign\expandafter\rising\else^{\underline{#1}}\fi}
\def\risingsign{^}
\def\rising#1{^{\overline{#1}}}

\title[Relative volume of comparable pairs under semigroup majorization]{Relative volume of comparable pairs under semigroup majorization}

\author[F.~D.~Cunden]{Fabio Deelan Cunden}
\address{Dipartimento di Matematica, Universit\'a degli Studi di Bari, I-70125 Bari, Italy, and
INFN, Sezione di Bari, I-70126 Bari, Italy}
\email{fabio.cunden@uniba.it}

\author[J. Czartowski]{Jakub Czartowski}
	\address{Doctoral School of Exact and Natural Sciences, Jagiellonian University, 30-348 Kraków, Poland}
	\address{Faculty of Physics, Astronomy and Applied Computer Science, Jagiellonian University, 30-348 Kraków, Poland}
        \address{School of Physical and Mathematical Sciences, Nanyang Technological University,
21 Nanyang Link, 637371 Singapore, Republic of Singapore}
 \email{jakub.czartowski@ntu.edu.sg}

\author[G.~Gramegna]{Giovanni Gramegna}
\address{Dipartimento di Fisica, Universit\`a degli Studi di Bari, I-70126 Bari, Italy, and
	INFN, Sezione di Bari, I-70126 Bari, Italy}
\email{giovanni.gramegna@uniba.it}

 \author[A. de Oliveira Junior]{A. de Oliveira Junior}
	\address{Center for Macroscopic Quantum States bigQ, Department of Physics,
Technical University of Denmark, Fysikvej 307, 2800 Kgs. Lyngby, Denmark}
\email{aleoli@dtu.dk}
\thanks{The research of FDC and GG is supported by Gruppo Nazionale di Fisica Matematica GNFM-INdAM and by Istituto Nazionale di Fisica Nucleare INFN through the project QUANTUM. FDC acknowledges the support from  PRIN 2022 project 2022TEB52W-PE1-
`The charm of integrability: from nonlinear waves to random matrices'.  GG acknowledges financial support from PNRR MUR project PE0000023-NQSTI and from the University of Bari through the 2023-UNBACLE-0245516 grant. AOJ acknowledges financial support from VILLUM FONDEN through a research Grant (40864) and EU Horizon Europe project QSNP (Grant Agreement No. 101114043).  
    JCz is supported by the start-up grant of the
    Nanyang Assistant Professorship at the Nanyang Technological University in Singapore, awarded to Nelly Ng.}

\begin{document}
 
\maketitle

\begin{abstract}  
Any semigroup $\SS$ of stochastic matrices induces a semigroup majorization relation $\prec^{\SS}$ on the set $\Delta_{n-1}$ of probability $n$-vectors. Pick $X,Y$ at random in $\Delta_{n-1}$: what is the probability that $X$ and $Y$ are comparable under $\prec^{\SS}$? We review recent asymptotic ($n\to\infty$) results and conjectures in the case of \emph{majorization} relation (when $\SS$ is the set of doubly stochastic matrices), discuss natural generalisations, and prove a new asymptotic result in the case of majorization, and new exact finite-$n$ formulae in the case of \emph{UT-majorization} relation, i.e.\ when $\SS$ is the set of upper-triangular stochastic matrices.
\end{abstract}

\section{Introduction}

The subject of combinatorial and algebraic aspects of partially ordered
sets has received much attention in recent decades in quantum physics, with an explosion
of works in the field, since the (re)discovery of the interplay between certain resource theories
and order relations, most notably between entanglement theory and majorization~\cite{nielsen2001majorization}. See~\cite{Du2015,chitambar2016comparison,RevModPhys.89.041003,deoliveirajunior2022} and references therein.
The basic idea of a (quantum) resource theory
is to study information processing under a restricted set of operations,  called \emph{free}. 
Thus, a resource theory leads to a partial ordering of states, ranking them from harder to easier to prepare by means of free operations.  
This notion of partial order, introduced in physics by Uhlmann~\cite{alberti1982stochasticity}, has deep mathematical roots dating back to Schur, Hardy, Littlewood, P\'olya, and others. 
\par

Consider a finite system whose states can be described by $n$-dimensional probability vectors $x=(x_1,\ldots,x_n)$, $x_i\geq0$, $\sum_ix_i=1$.  These states belong to the probability simplex $\Delta_{n-1}$ referred to as the \textbf{state space}. Stochastic matrices represent the most general linear transformations between states in~$\Delta_{n-1}$. Formally, a {resource theory}~\cite{chitambar2019quantum} can be constructed by declaring that the set of {\textbf{free operations}} between states in $\Delta_{n-1}$ is a given proper semigroup $\mathcal{S}$ of stochastic matrices. The set of {free operations} $\mathcal{S}$ leads to an order relation in the state space: we write $x\prec^{\mathcal{S}} y$ whenever $x=yP$ for some $P\in\mathcal{S}$, that is if the state $y$ can be converted into $x$ by means of a {free operation}. In this paper, we adopt the right (row) stochastic matrix convention, where stochastic matrices act on the right on row probability vectors. This is standard in probability theory, where stochastic matrices are also called Markov matrices (see Section~\ref{sec:background}).  The standard majorization relation~\cite{MOA} is obtained by choosing $\mathcal{S}$ as the semigroup of doubly stochastic matrices. An equivalent characterization of majorization (and other semigroup majorization) is in terms of monotonicity of family of functions that deserve the name of `generalised entropies'~\cite{renyi1961proceedings}. In some situations the family of generalised entropies can be generated by nonnegative linear combinations of a finite family of `monotone' functions.

Given two states $x,y$, for any {free operations} $P\in \mathcal{S}$, one can certainly decompose $yP$,  as a convex combination $px+(1-p)z$ for some $z\in\Delta_{n-1}$ (with $p$ that depends on $P$). The trivial case $p=1$ corresponds the case $yP=x$, i.e., $x\prec^{\mathcal{S}}y$. For general $x,y$ it is natural to ask what is the largest value of $p$, denoted with $\Pi_{\mathcal{S}}(x,y)$, that can be achieved in a transformation $yP=px+(1-p)z$, as $P$ runs over all {free} transformations. In the context of resource theories $\Pi_{\mathcal{S}}$ can be thought of as the \emph{maximal success probability of conversion} of $y$ into $x$~\cite{vidal1999entanglement,Du2015,zhu2017operational,oszmaniec2019operational}. This happens when a quantum measurement is used in the conversion protocol, which allows one to extract $x$ from the convex combination $px+(1-p)z$ conditionally on  certain outcomes of the measurement, whose occurrence probability sum to $p$. However, it is worth mentioning that measurements may not be free operations in a given quantum resource theory~\cite{alhambra2016fluctuating}.
\par
 
It is clear that in a resource theory, a state $y$ is `useful' (or `more useful than others') if from $y$ it is possible to reach `many' other states by {free operations}. For example, in entanglement theory maximally entangled states of a bipartite system can be converted to \emph{every} other state using local operations and classical communication only, and this is why such special states (like Bell states), sometimes referred to as `golden state', are key ingredients in applications. But such golden
states are very special.
\par

We therefore address the following question:
{given a fixed resource theory, how}
likely is it that two `generic' states can be connected via a free operation in $\mathcal{S}$? More precisely, pick $X$ and $Y$ at random according to a (natural) probability measure on the state space $\Delta_{n-1}$: what is the probability that $X\prec^{\mathcal{S}} Y$? Our {initial} intuition is that in the limit of large dimensionality $n\to\infty$,  typical pairs of states are so uncorrelated, that {they} cannot be converted to each other via free operations and hence are incomparable. This \textit{`comparability is atypical'} statement was first put forward by Nielsen~\cite{Nielsen1999} in quantum theory and was only recently proved in~\cite{Cunden2021} and in~\cite{Jain24} for the resource theories of quantum coherence and quantum entanglement, respectively. Given that exact conversion between generic states is typically impossible to achieve by means of free operations, one can consider transformations that succeed with some probability $p$ (not necessarily equal to $1$) at transforming the initial state into the desired final state.  If we allow nondeterministic transformations (i.e., some probability of failure), what is the typical maximal probability of successful transformation $\Pi_{\mathcal{S}}(X,Y)$ for random states $X$ and $Y$?
\par
\subsection{Majorization} Let $\prec$ be the majorization relation. This is the semigroup majorization corresponding to the set of doubly stochastic matrices (precise definitions are given later). Let also $\Pi(x,y)$ be the largest value $p\in[0,1]$ such that $px+(1-p)z\prec y$ for some $z\in\Delta_{n-1}$.
\begin{mainthm}
\label{thm:intro_maj}
Pick $X,Y$ independently and uniformly at random in $\Delta_{n-1}$. Then,
\be
\label{eq:main1_maj}
\lim_{n\to\infty}P\left(X\prec Y \right)=0,
\ee
and $\Pi(X,Y)$ converges in probability of $1$: for all $\epsilon>0$,
\be
\label{eq:main2_maj}
\lim_{n\to\infty}P(\Pi(X,Y)< 1-\epsilon)=0.
\ee
\end{mainthm}
The limit~\eqref{eq:main1_maj} is a restatement of a claim proved in~\cite{Cunden2021}, while the convergence~\eqref{eq:main2_maj} is a new result (Theorem~\ref{thm:limit} below). 
Here is an informal rephrasing: the proportion of pairs of states $(x,y)$ such that $x$ is majorised by $y$ tends to $0$ for large $n$; nevertheless in any $\epsilon$-neighbourhood of $x$, if $n$ is large, there is a vector $x_{\epsilon}$ that is majorized by $y$. The result may, at first,  feel paradoxical. A way to
reconcile these facts is to observe that~\eqref{eq:main1_maj}-\eqref{eq:main2_maj} show that
exchanging limits  does not work here: for fixed $\epsilon>0$, the $n\to\infty$ limit yields $1$, but if
one first takes $\epsilon\to0$ and then $n\to\infty$, then the probability becomes zero.
\subsection{Upper-triangular majorization} Let $\prec^{\UT}$ be the semigroup majorization relation corresponding to  the semigroup $\mathcal{S}^{\UT}$ of upper-triangular stochastic matrices (UT-majorization). Let also $\Pi_{\UT}(x,y)$ be the largest value $p\in[0,1]$ such that $px+(1-p)z\prec^{UT}y$ for some $z\in\Delta_{n-1}$.
In physics, the semigroup $\mathcal{S}^{\UT}$  arises in the zero-temperature limit of thermal operations~\cite{narasimhachar2015low}.
\begin{mainthm}
\label{thm:intro}
Pick $X,Y$ independently and uniformly at random in $\Delta_{n-1}$. Then,
\be
\label{eq:main1}
P\left(X\prec^{\UT}Y\right)=\frac{1}{n}.
\ee
Moreover, the distribution function of $\Pi_{\UT}(X,Y)$ is
\be
\label{eq:main2}
P(\Pi_{\UT}(X,Y)\leq t)=
\begin{cases}
0&\text{if $t\leq0$,}\\
\left(1-\dfrac{1}{n}\right)t&\text{if $0<t<1$,}\\
1&\text{if $t\geq1$}.
\end{cases}
\ee
\end{mainthm}
A more precise statement of~\eqref{eq:main1}  (for more general measures $\mu_n$ on $\Delta_{n-1}$) and~\eqref{eq:main2} is given in Theorems~\ref{thm:VanishingProb} and~\ref{thm:cumulativeDist}  later on. Eq.~\eqref{eq:main1} implies the claimed `comparability is atypical' statement: the probability that two random states can be connected by a upper-triangular stochastic matrix vanishes in the limit $n\to\infty$. Eq.~\eqref{eq:main2} addresses the question of `nondeterministic convertibility'. 
\par

We stress that Theorem~\ref{thm:intro} is an exact finite-$n$ result, and in this sense the model of comparability of random vectors with respect to UT-majorization can be considered `exactly solvable':  explicit calculations are possible and might give insights to understand the asymptotics of other, non-solvable models. 

\subsection{Outline of the paper} In the next section, we give some background on partial orders, semigroup majorization, and orderings induced by function sets. In Section~\ref{sec:comp}, we formulate the problem, recall known asymptotic results for majorization (Theorems~\ref{thm:CFFG21} and~\ref{thm:JKM24}) and state the first main theorem on UT-majorization,  Theorem~\ref{thm:VanishingProb}, which gives an exact universal formula for the probability that two random states $X$ and $Y$ are comparable with respect to UT-majorization. In Section~\ref{sec:weak} we consider the maximal success probability of conversion of random states $X$ and $Y$. In the case of majorization and UT-majorization relations, we first present and prove an explicit formula for $\Pi_{\mathcal{S}}(\cdot,\cdot)$ involving ratios of monotone functions (Propositions~\ref{prop:VidalMaj} and~\ref{prop:VidalUT}). In the majorization relation, one expects the existence of a limit in distribution of the random variable $\Pi_{\mathcal{S}}(X,Y)$. We prove that this is indeed the case and we find the limit distribution to be concentrated at $1$ (Theorem~\ref{thm:limit}). For UT-majorization, we find the exact finite-$n$ distribution function (Theorem~\ref{thm:cumulativeDist}), which, in fact, has a nontrivial limit. The final section contains the proofs of the main results.
We included some technical results in the appendices~\ref{app:Bolshev} (Bolshev's recursion) and~\ref{app:stick} (stick-breaking representation of the Dirichlet distribution). 

\section{Background}
\label{sec:background}
 We begin by recalling some terminology. The reader is referred to the monograph of Marshall, Olkin, and Arnold~\cite{MOA} for further details.

\subsection{Order relations}  
A \textbf{partial order} is a pair $\left(\AA,\sqsubseteq\right)$, where $\AA$ is a set and $\sqsubseteq$ a binary relation on $\AA$ satisfying
\begin{enumerate}
\item \textbf{reflexivity:} for all $x\in\AA$, $x\sqsubseteq x$.
\item \textbf{transitivity:} for all $x,y,z\in\AA$, if $x\sqsubseteq y$ and $y\sqsubseteq z$, then $x\sqsubseteq z$.
\item \textbf{antisymmetry:} for all $x,y\in\AA$, if $x\sqsubseteq y$ and $y\sqsubseteq x$, then $x=y$.
\end{enumerate}
A \textbf{preorder} is a pair $\left(\AA,\sqsubseteq\right)$  satisfying the first two conditions. If $x\sqsubseteq y$ we say that `$x$ {is less than} $y$'. 
If it exists, a \textbf{bottom element} of a partial order, is an element $\bot\in\AA$ satisfying $\bot\sqsubseteq x$ for all $x\in\AA$. If it exists, a \textbf{top element} of a partial order, an element $\top\in\AA$ satisfying $x\sqsubseteq \top$ for all $x\in\AA$. If neither $x\sqsubseteq y$ nor $y\sqsubseteq x$, we say that $x$ and $y$ are \textbf{incomparable}.
A \textbf{total order} is a partial order $\left(\AA,\sqsubseteq\right)$ in which for every pair of elements $x,y$ either $x\sqsubseteq y$ or $y\sqsubseteq x$ holds. The \textbf{dual} (or converse) relation $\left(\AA,\sqsubseteq^*\right)$, is defined by setting $x \sqsubseteq^* y$ if and only if $y \sqsubseteq x$. 

For a given $y\in\AA$, the set of all \textbf{antecedents} of $y$ in the relation $\sqsubseteq$ is denoted
by $\gamma(y,\sqsubseteq)$. Thus
   \be
   \gamma(y,\sqsubseteq)=\{x\in\AA\colon x\sqsubseteq y\}.
   \ee
      The \textbf{graph} of the relation $\sqsubseteq$ on $\AA$ will be denoted by $\Gamma(\AA,\sqsubseteq)$. Thus
   \be
\Gamma(\AA,\sqsubseteq) = \{(x,y)\in \AA\times \AA\colon x\in\gamma(y,\sqsubseteq)\}= \{(x,y)\in \AA\times \AA\colon x\sqsubseteq y\}.
\ee
We will simply write $\Gamma$ when $\left(\AA,\sqsubseteq\right)$  is clear from the context.

\subsection{Semigroup majorization}
 For $x\in\R^n$, we denote by $x^\downarrow$ the permutation of $x$ with the components arranged in nonincreasing order, i.e., $x_1^\downarrow \geq x_2^\downarrow \geq \cdots \geq x_n^\downarrow$. For $x,y\in\R^n$, we say that $x$ is \emph{majorized} by $y$ (or $y$ majorizes $x$) and write $x\prec y$ if
   \be
   \begin{cases}
 \displaystyle  \sum_{i=1}^{k} x^{\downarrow}_i\leq    \sum_{i=1}^{k} y^{\downarrow}_i, &\text{for $k=1,\ldots,n-1$}\\\\
\displaystyle   \sum_{i=1}^{n} x^{\downarrow}_i=    \sum_{i=1}^{n} y^{\downarrow}_i.
   \end{cases}\label{eq:majorization}
   \ee
    If the equality in \eqref{eq:majorization} is replaced by a corresponding inequality, we say that $x$ is \emph{weakly  (sub)majorized} by $y$ and write $x\prec_w y$, i.e., for $x,y\in\R^n$:
   \be
   x\prec_{w} y\quad\text{if}\quad
 \displaystyle  \sum_{i=1}^{k} x^{\downarrow}_i\leq    \sum_{i=1}^{k} y^{\downarrow}_i, \quad\text{$k=1,\ldots,n$}.
   \ee
  
    It can be seen~\cite{Malamud} that   $x\prec y$ if and only if, for $j =1,\ldots,n$,
   \begin{multline}
   \label{eq:conv_hull_i}
   \conv\{x_{i_1}+\cdots+x_{i_j}\colon 1\leq i_1< \cdots < i_j\leq n\}
 \subset     \conv\{y_{i_1}+\cdots+y_{i_j}\colon 1\leq i_1<\cdots< i_j\leq n\},
   \end{multline}
where $\conv(A)$ denotes the convex hull of a set $A$ (the smallest convex set containing $A$).
\par

   By a celebrated result by Hardy, Polya and Littlewood~\cite{HLP}, $x \prec y$ if and only if $x=yD$ for some doubly stochastic matrix $D$. This relation is a preorder because the set of $n\times n$ doubly stochastic matrices contains the identity and is closed under multiplication. Other preorders are obtained if the set of doubly stochastic matrices are replaced by another semigroup $\SS$ of matrices with identity~\cite{MOA}.
\par

Let $\AA$ be a subset of $\R^n$ and let $\SS$ be a semigroup of linear transformations (matrices) mapping $\AA$ to $\AA$. A vector  $x\in\AA$  is said to be \textbf{semigroup majorized} by $y$, written $x\prec^{\SS} y$, if $x = yS$ for some $S\in \SS$. From the semigroup property of $\SS$, it follows that $\left(\AA,\prec^{\SS}\right)$ is a preorder, but not a partial order in general. See~\cite{joe1990majorization,vomEnde,ende2022thermomajorization} for further details.
\par

Recall that a square matrix $P$ is row (sub)stochastic if all entries are nonnegative and the row sums are (less than or) equal to $1$. If we denote the row vector $e=(1,\ldots,1)$, we have that a square matrix $P$ with nonnegative entries is row  (sub)stochastic if $Pe'=e'$ ($Pe'\leq e'$), where $e'$ denotes the transposed vector of $e$ and $\leq$ is to be understood component-wise.
\par

We define the following {types} of matrices with nonnegative entries (mapping $\mathcal{A}=\R_+^n$ to itself):
\begin{itemize}
\item \textbf{Doubly stochastic matrices (DSM)}: \\
row stochastic matrices $D$ with column sums equal to $1$, i.e., $eD=e$.
\item \textbf{Weakly doubly stochastic matrices (wDSM)}: \\
row substochastic matrices $D$ with column sums at most $1$, i.e., $eD\leq e$.
\item \textbf{Upper triangular stochastic matrices (UTSM)}: \\
row stochastic matrices $S=(s_{ij})$ with $s_{ij}=0$ whenever $i>j$.
\item \textbf{Weakly upper triangular stochastic matrices (wUTSM)}: \\
row substochastic matrices $S=(s_{ij})$ with $s_{ij}=0$ whenever $i>j$.
\end{itemize}

The sets of $n\times n$ DSMs, wDSM, UTSMs, and wUTSMs,  are all semigroups with identity.
We denote the semigroup majorization of DSMs, wDSM, UTSMs, and wUTSMs by 
$\prec$, $\prec_w$, $\prec^{\UT}$, $\prec^{\UT}_w$ for short, and we say that $x$ is majorized, weakly majorized,  UT-majorized, weakly UT-majorized by $y$, whenever $x\prec y$, $x\prec_{w} y$, $x\prec^{\UT} y$, or $x\prec^{\UT}_w y$, respectively.
\par

    Birkhoff's theorem~\cite{Birkhoff} states that the set of doubly stochastic matrices is the convex hull of all permutations matrices. As a consequence, $x\prec y$ if and only if $x$ lies in the convex hull of the $n!$ permutations of $y$ 
   (this is a restatement of~\eqref{eq:conv_hull_i}).
     The set of $n\times n$ doubly substochastic matrices is the convex hull of the set of $n\times n$ matrices that have at most one unit in each row and each column, and all other entries are zero~\cite{Mirsky}. 
Therefore  $x\prec_w y$ if and only if $x$ is in the convex hull of the  points $(\eta_1y_{\pi(1)},\ldots,\eta_ny_{\pi(n)})$, where $\pi$ is a permutation and each $\eta_i$ is $0$ or $1$.
\par

Let $\SS^{\UT}$ and $\SS_w^{\UT}$  be the sets of  UTSMs and wUTSMs, respectively. (all entries are nonnegative and the row sums are equal to $1$). It is easy to verify that both $\SS^{\UT}$ and $\SS_w^{\UT}$ are convex sets. The extremal points of $\SS^{\UT}$ are upper-triangular matrices with a single entry $1$ in each row and all other entries zero. The extremal points of $\SS_w^{\UT}$ are upper-triangular matrices with at most one entry $1$ in each row and all other entries zero.
\par

To summarise, the {antecedents} of $y\in\Delta_{n-1}$ in $\mathcal{A}=\R_+^n$ for the preorders $\prec$, $\prec_w$, $\prec^{\UT}$ and $\prec_{w}^{\UT}$ are
\be
\gamma(y,\prec)=\conv\{\left(y_{\pi(1)},\ldots,y_{\pi(n)}\right)\colon \pi\in S_n\},
\ee
\be
\gamma(y,\prec_w)=\conv\{\left(\eta_1y_{\pi(1)},\ldots,\eta_ny_{\pi(n)}\right)\colon \pi\in S_n, \eta_i\in\{0,1\}\},
\ee
\be
\gamma(y,\prec^{\UT})
=\conv\left\{\left(\sum_{ j_1\in f^{-1}(1)}y_{j_1},\ldots,\sum_{j_n\in f^{-1}(n)}y_{j_n}\right)\colon f\colon[n]\to[n]\,\,\text{s.t. $f(i)\geq i$ for all $i$ }\right\}.
\ee
\begin{multline}
\gamma(y,\prec_w^{\UT})
=\conv\Biggl\{\left(\sum_{ j_1\in f^{-1}(1)}\eta_{j_1}y_{j_1},\ldots,\sum_{j_n\in f^{-1}(n)}\eta_{j_n}y_{j_n}\right)\colon \\f\colon[n]\to[n]\,\,\text{s.t. $f(i)\geq i$ for all $i$, and } \eta\in\{0,1\}^n\Biggr\}.
\end{multline}

\begin{ex}
For $n=3$, let $y=(y_1,y_2,y_3)$. Then,
\begin{multline}
\gamma(y,\prec)=\conv\{\left(y_1,y_2,y_3\right),\left(y_2,y_1,y_3\right),\left(y_2,y_3,y_1\right),
\left(y_1,y_3,y_2\right),\left(y_3,y_1,y_2\right),\left(y_3,y_2,y_1\right)\},
\end{multline}
\begin{multline}
\gamma(y,\prec^{\UT})=\conv\{\left(y_1,y_2,y_3\right),\left(y_1,0,y_2+y_3\right),\left(0,y_1+y_2,y_3\right),\\\left(0,y_2,y_1+y_3\right),\left(0,y_1,y_2+y_3\right),\left(0,0,y_1+y_2+y_3\right)\}.
\end{multline}
See Fig.~\ref{F-convexulls} for examples of the antecedent sets.
\end{ex}
\begin{rem}
    The number of extremal points of $\gamma(y,\prec)$ is generically $n!$ (the number of permutations of an $n$-set); the number of extremal points of $\gamma(y,\prec_w)$ is generically $\sum_{k=0}^nk!\binom{n}{k}^2$, the number of 
 partial permutations of an $n$-set~\cite{OEISA002720}. These are upper bounds for the number of extremal points of $\gamma(y,\prec^{\rm UT})$ and $\gamma(y,\prec_w^{\rm UT})$, respectively.
See Fig.~\ref{F-convexulls}.
\end{rem}
\begin{figure*}
	\includegraphics{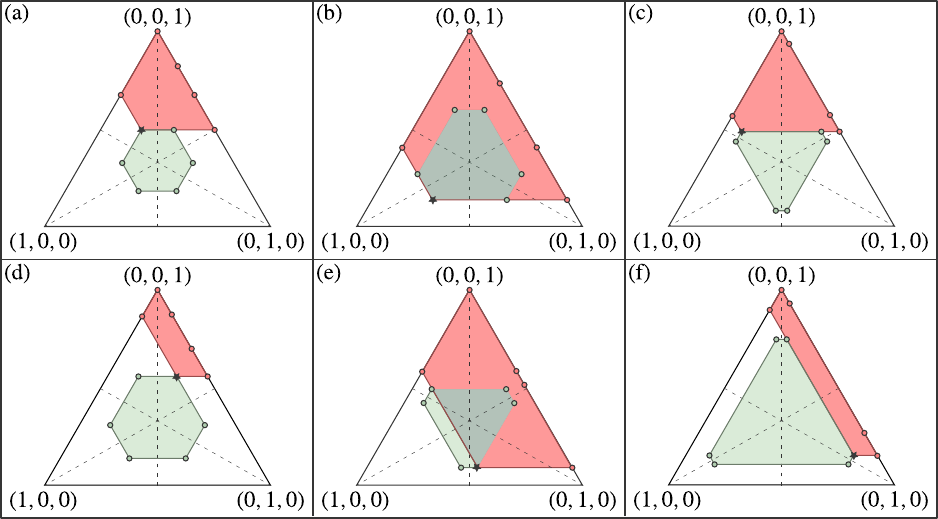}
        \captionsetup{width=1\linewidth}  
	\caption{Antecedent sets $\gamma(y, \prec)$ (in green) and $\gamma(y, \prec^{\UT})$ (in red) for each $y \in \Delta_2$ (depicted as a black star), where $y$ is (a) $y = (0.33, 0.18, 0.49)$, (b) $y = (0.60, 0.26, 0.14)$, (c) $y = (0.43, 0.08, 0.49)$, (d) $y = (0.14, 0.31, 0.55)$, (e) $y = (0.42, 0.49, 0.09)$, and (f) $y = (0.10, 0.74, 0.16)$. Points obtained by applying extreme operations with respect to majorization and UT-majorization are depicted by green and red circles, respectively. Note that these are not necessarily extreme points of the corresponding antecedent sets.}
	\label{F-convexulls}
\end{figure*}
It can be verified~\cite{ParkerRam} that UT-majorization is equivalent to unordered majorization, which is defined by the same inequalities \eqref{eq:majorization} where the components of the vectors are unordered. To avoid further notation, we still denote unordered majorization with  $x\prec^{\UT} y$, so that we have:
   \be
   x\prec^{\UT} y\quad\text{iff}\quad
   \begin{cases}
 \displaystyle  \sum_{i=1}^{k} x^{}_i\leq    \sum_{i=1}^{k} y^{}_i, &\text{$k=1,\ldots,n-1$}\\\\
\displaystyle   \sum_{i=1}^{n} x^{}_i=    \sum_{i=1}^{n} y^{}_i.
   \end{cases}
   \ee
Note that if $x=yM$, with $M\in \SS^{\UT}$, then $x_i=\sum_{j=1}^ny_jM_{ji}=\sum_{j=1}^iy_jM_{ji}$ since $M_{ji}=0$ for $j>i$. Moreover, $\sum_{j=i}^nM_{ij}=1$. Note that $M_{nn}=1$ and $e_n=e_nM$ for all $M\in \SS^{\UT}$.

\begin{rem}[Physics] Another prominent example of semigroup majorization is obtained from the semigroup $\SS^q$ of \textbf{$q$-stochastic matrices}  defined as follows.
Let $q$ be a nonnegative vector.  A row stochastic matrix $S$ is said to be $q$-stochastic if 
$qS=q$.
For $x,y\in\R^{n}$ we say that $x$ is $q$-majorized by $y$, and write $x\prec_{ \SS^{q}} y$ if $x=yS$ for some $S\in \SS^{q}$. In $q$-majorization, $q=\bot$ is always the unique bottom element, while $e_k$ is a top element if and only if $q_k$ is minimal in $q$. See~\cite[Theorem 3]{vomEnde}. For $q=e$ we get the set of DSMs. 
 Notice that $\SS^{\UT}$ is a proper subsemigroup of $\SS^{e_n}$,  $e_n=(0,\ldots,0,1)$. 
 \par

If we write $q=\left(e^{-\beta E_1},\ldots,e^{-\beta E_n}\right)$ for some sequence of energy levels $E_1>\cdots> E_n\in \R$ and inverse temperature $\beta>0$, then $q$-majorization is equivalent to \textbf{thermomajorization}~\cite{horodecki2013fundamental,shiraishi2020two}, as $\SS^q$ is the semigroup that fixes the `Gibbs state' $q$. At high temperature $\beta\to0$, thermomajorization reduces to standard majorization; at zero temperature $\beta\to\infty$, thermomajorization relation collapses to UT-majorization. Note that in physics it is customary to order the energy levels in increasing order $E_1<\cdots<E_n$. To make a correspondence with the existing literature~\cite{horodecki2013fundamental,narasimhachar2015low} it is enough, by taking the transpose, to consider column-stochastic (rather than row-stochastic) matrices acting on column probability vectors. Note that transposing an upper-triangular matrix yields a lower-triangular matrix.
\end{rem}
\begin{rem}[Combinatorics]
In the realm of discrete combinatorics, the counterpart of majorization is called \textbf{dominance order} for integer partitions or integer compositions. There exists a natural extension due to Narayana~\cite{Narayana59} that can be used in the continuous setting to provide yet another generalisation of majorization. Let $s\in\R$. For $x,y\in\R^n$, the vector $x$ is said to be \emph{$s$-dominated} by $y$, and we write $x\prec^s y$ if
  \be
 \displaystyle  \sum_{i=1}^{k} x^{\downarrow}_i-    \sum_{i=1}^{k} y^{\downarrow}_i\leq s, \text{for $k=1,\ldots,n$}.\label{eq:sdominance}
   \ee
When $s=0$, the $s$-dominance order is the weak majorization order. 
\end{rem}
\par

\subsection{Orderings induced by function sets}
Here is another convenient description of  partial orders induced by functions that can be thought of as generalised entropies.
Let $\Phi$ be a set of real-valued functions $\phi$ defined on a subset $\mathcal{A}$ of $\R^n$. For $x,y\in\mathcal{A}$ we write
\begin{equation}\label{eq:majPhi}
	x\prec_\Phi y\quad \text{ iff }\quad \phi(x)\leq \phi(y),\, \text{ for all }\phi\in\Phi.
\end{equation}
If \eqref{eq:majPhi} holds, then it holds also for all linear combinations of functions in $\Phi$ with positive coefficients. These linear combinations form a cone that is said to be generated by  $\Phi$. 
\par

For instance, let $\Phi:=\{\phi_1,\ldots,\phi_{n+1}\}$, where 
	\begin{equation}\label{eq:frameMaj}
		\phi_k(x)=\sum_{i=1}^{k}x_i^\downarrow, \quad k=1,\dots,n, \qquad \phi_{n+1}(x)=-\phi_n(x).	\end{equation}
	Then $\prec_{\Phi}$ is equivalent to $\prec$ according to definition \eqref{eq:majorization}. 
The cone generated by  $\Phi_w:=\{\phi_1,\ldots,\phi_n\}$, with $\phi_k$ in~\eqref{eq:frameMaj} induces the weak majorization preorder on $\mathbb{R}_+^n$.
\par

Similarly, the sets $\Phi^{\UT}_w:=\{\phi_1^{\UT},\dots,\phi_{n}^{\UT}\}$, and  $\Phi^{\UT}:=\Phi^{\UT}_w\cup\{\phi_{n+1}^{\UT}\},$
where
\begin{equation}
\label{eq:frameUT}
	\phi_k^{\UT}(x)=\sum_{i=1}^k x_i, \quad k=1,\dots, n, \qquad \phi_{n+1}^{\UT}(x)=-\phi_n^{\UT}(x),
\end{equation}
generate the cones inducing weak UT-majorization and UT majorization, respectively.

\section{Probability that generic states are comparable}
\label{sec:comp}
Let $(\AA,\sqsubseteq)$ be a preorder with $\AA\subset \R^n$,
and let $\mu$ be a probability measure on $\AA$. If $\Gamma$ is measurable, we define the \textbf{probability} of $\Gamma$ as
\be
P(\Gamma)=\mu{\otimes}\mu\left(\Gamma\right)=\mu{\otimes}\mu\left(\{(x,y)\in \AA\times \AA\colon x\sqsubseteq y\}\right).
\ee
The above quantity can be interpreted as follows. Pick $X,Y$ independently at random according to the probability measure $\mu$. Then $P(\Gamma)=P(X\sqsubseteq Y)$ is the probability that $X$ is less than $Y$. By symmetry, this is also equal to the probability $P(Y\sqsubseteq X)$, since $X,Y$ are independent and identically distributed.  
\par

In this paper we consider the case of partial orders on the simplex of discrete probability vectors. 

Let $n\in\N$ and consider the $(n-1)$-dimensional simplex 
\be
	\Delta_{n-1}:=\left\{x\in\R^{n}\colon x_i\geq0, \sum_{i=1}^n x_i=1 \right\}=\conv\{e_1,\ldots,e_n\},
\ee
where $e_1,\ldots,e_n$ stand for the unit vectors of the
standard orthonormal basis of $\R^n$.
\par

With this notation, consider the poset $\left(\AA,\sqsubseteq\right)=\left(\Delta_{n-1},\prec\right)$, with 
$\prec$ the preorder of majorization. 
Note that the bottom element is  $\bot=\frac{1}{n}(e_1+\cdots+e_n)=(1/n,\ldots,1/n)$, while the top elements are $e_1,\ldots, e_n$.
In~\cite{Cunden2021}, Cunden, Facchi, Florio and Gramegna,  proved the following asymptotic result.
\begin{thm}[\cite{Cunden2021}]\label{thm:CFFG21}
Let $\Gamma_n=\Gamma\left(\Delta_{n-1},\prec\right)$, and $\mu_n$ the uniform probability measure on the simplex $\Delta_{n-1}$. Write $P(\Gamma_n)=\mu_n\otimes \mu_n(\Gamma_n)$. Then,
\be
\lim_{n\to\infty}P(\Gamma_n)=0.
\ee
\end{thm}
This statement is Eq.~\eqref{eq:main1_maj} of Theorem~\ref{thm:intro_maj}. The proof of Theorem~\ref{thm:CFFG21} in~\cite{Cunden2021} uses Kolmogorov's 0-1 law and hence it is fundamentally
non-quantitative.  
Recently, Jain,  Kwan and Michelen~\cite{Jain24} proved the following quantitative (though presumably not optimal) estimate.
\begin{thm}[\cite{Jain24}]
\label{thm:JKM24}
With the same notation of Theorem~\ref{thm:CFFG21}, there is a constant $C>0$ such that
\be
P(\Gamma_n)\leq C\left(\log n/n\right)^{1/16}.
\ee
\end{thm}
The same authors~\cite{Jain24}, proved the analogue of Theorem~\ref{thm:CFFG21} when the measure $\mu_n$ on $\Delta_{n-1}$ comes from the  fixed-trace Wishart measure of random matrix theory, thus settling a conjecture put forward by Nielsen~\cite{Nielsen1999} and numerically investigated in~\cite{cunden2020} in the context of entanglement theory. 
\par

 Consider now the poset $\left(\AA,\sqsubseteq\right)=\left(\Delta_{n-1},\prec^{\UT}\right)$; in this case the bottom element is  $\bot=e_n=(0,0,\ldots,1)$, while the top element is $\top=e_1=(1,0,\ldots,0)$.
 We state the following simple theorem on the exact probability of $\Gamma_n$ for UT-majorization. 
\begin{thm}\label{thm:VanishingProb}  Consider the graph $\Gamma_n=\Gamma\left(\Delta_{n-1},\prec^{{\UT}}\right)$ and let $\mu_n$ be an absolutely continuous and permutation symmetric probability measure on $\Delta_{n-1}$. Then, for all $n\geq2$, 
\be
P(\Gamma_n)=\frac{1}{n}.
\ee
\end{thm}
 When $\mu_n$ is the uniform measure on the simplex, Theorem~\eqref{thm:VanishingProb} provides  Eq.~\eqref{eq:main1} of Theorem~\ref{thm:intro}.
 
 \section{Maximal success probability of conversion}
 \label{sec:weak}

Note that for $x,y\in\Delta_{n-1}$, $x\prec y$ if and only if $x\prec_w y$, since the normalization condition automatically guarantees the equality constraint in Eq.~\eqref{eq:majorization}. 

For $x,y\in\Delta_{n-1}$, it can happen that $x$ is not majorized by $y$, but it is always possible to find $p\in[0,1]$ and $z\in \Delta_{n-1}$ such that
\be
\label{eq:convexMaj}
px+(1-p)z\prec y.
\ee
Let
\be
\Pi:\Delta_{n-1}\times \Delta_{n-1}\to[0,1],
\ee
with $\Pi(x,y)$ defined as the largest $p\in[0,1]$ such that~\eqref{eq:convexMaj} holds for some $z\in\Delta_{n-1}$. Equivalently, $\Pi(x,y)$ is the largest $p$ such that $yD=px+(1-p)z$ for some DSM $D$ and some vector $z\in\Delta_{n-1}$. 
\par

The function $\Pi$ can be expressed in terms of the monotones  $\phi_1,\ldots,\phi_{n}$  as defined in \eqref{eq:frameMaj}.
\begin{prop}
\label{prop:VidalMaj}
\be
\label{eq:propVidalMaj}
	\Pi(x,y)=\min_{1\leq k \leq n}\frac{\phi_k(y)}{\phi_k(x)}.
	\ee
	where $\phi_k$ are defined in~\eqref{eq:frameMaj}.\
\end{prop}
\begin{proof}
Let $p^*=\min_{1\leq k \leq n}\frac{\phi_k(y)}{\phi_k(x)}$. 
We first show that
\begin{equation}\label{eq:Pi}
p^*=\max\{p\in[0,1]:px\prec_w y\}.
\end{equation}
 Let $k^*\in\{1,\ldots,n\}$ such that $p^*=\frac{\phi_{k^*}(y)}{\phi_{k^*}(x)}$. Then,
$$
\phi_k(p^*x)=p^*\phi_k(x)=\left(\frac{\phi_{k^*}(y)}{\phi_{k^*}(x)}\right)\phi_k(x)\leq \left(\frac{\phi_{k}(y)} {\phi_{k}(x)}\right)\phi_k(x)=\phi_k(y), \, \text{for all $k=1,\ldots,n$}.
$$
Hence $p^*x\prec_{w}y$. Now suppose that there is $p'>p^*$ such that $p'x \prec_{w}y$. Then $\phi_k(p'x)>\phi_k(p^*x)$ for all $k=1,\ldots,n$. Choosing $k=k^*$,
$$
\phi_{k^*}(p'x)>\phi_{k^*}(p^*x)=p^*\phi_{k^*}(x)=\phi_{k^*}(y).
$$
 we get a contradiction. Hence $p^*$ is the largest value $p\in[0,1]$ such that $px\prec_w y$, i.e., such that
 $
 yP=px
 $
 for some wDSM $P$.
 \par
 
 We can now prove the claim~\eqref{eq:propVidalMaj}. The case $p^*=1$ is trivial. Assume $p^*<1$. Reasoning as above, we first show that $\Pi(x,y)\leq p^*$. Suppose by contradiction that there is $p'>p^*$ such that $p'x+(1-p')z\prec y$, for some $z\in\Delta_{n-1}$. Then, $\phi_k(y)\geq \phi_k(p'x+(1-p')z)\geq \phi_k(p'x) >\phi_k(p^*x)$ for all $k=1,\ldots,n$ and, choosing $k=k^*$, we get a contradiction. 
 \par 
 
 It remains to show that there exists a DSM $D$ such that $yD=p^*x+(1-p^*)z$ with $z\in\Delta_{n-1}$.  By~\eqref{eq:Pi} we know that $yP=p^*x$ with $P$ wDSM. A theorem attributed to von Neumann (see~\cite[P. 37]{MOA}) ensures that there exists a DSM $D$ such that $P_{ij}\leq D_{ij}$ for all $i,j=1,\ldots,n$. Note that $P':=D-P$ is also a wDSM. Then,
 \be
 yD=yP+y(D-P)=p^*x+(1-p^*)z,
 \ee
 where $z=\frac{1}{1-p^*}y(D-P)\in\Delta_{n-1}$.
\end{proof}
\begin{figure}
 \includegraphics[width=\textwidth]{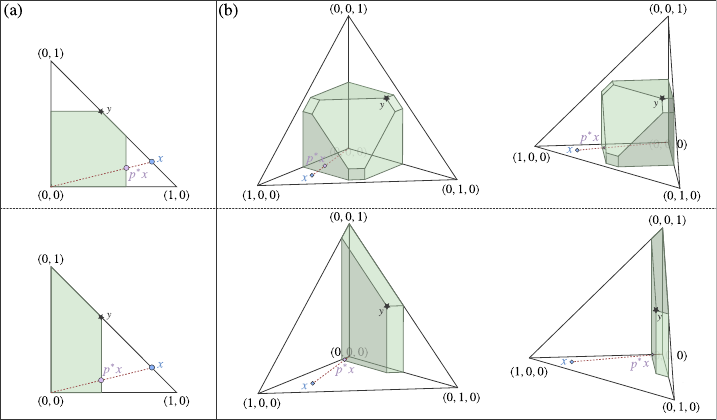}
       \captionsetup{width=1\linewidth}  
	\caption{The antecedent set of $y$ is shown for both the weak (top panel) and weak UT (bottom panel) majorizations for (a) $x = \left(\frac{8}{10}, \frac{2}{10}\right)$ (blue circle) and $y = \left(\frac{2}{5}, \frac{3}{5}\right)$ (black star), with $p^{*}x$ (pink circle) lying on the boundary of the antecedent set, and (b) $x = \left(\frac{7}{10}, \frac{1}{4}, \frac{1}{20}\right)$ (blue circle) and $y = \left(\frac{1}{12}, \frac{5}{12}, \frac{1}{2}\right)$ (black star), again with $p^{*}x$ (pink circle) lying on the boundary of the antecedent set.}
	\label{fig:example-3}
\end{figure}

Figure~\ref{fig:example-3} illustrates Proposition 1 by showing how the optimal value $p^*$ of $p$ in the expression $p x + (1 - p) z \prec y$ corresponds to the boundary of the antecedent set of $y$ under  weak majorization (top panel) and  weak UT-majorization (bottom panel). In both cases, the point $p^* x$ (shown in pink) lies on the boundary of the shaded region, confirming that it is the largest value for which $p x$ alone remains weakly majorized (UT-majorized) by $y$. See~\eqref{eq:Pi}. 

Note that
\be
1/n\leq \Pi(x,y)\leq 1.
\ee
The minimum is achieved for $x=e_1=(1,0,\dots,0)$ and $y=(1/n,\ldots,1/n)$.
On the other hand, $\Pi(x,y)=1$ if and only if $x\prec y$, and so 
\be
\Gamma(\Delta_{n-1},\prec)=\{(x,y)\in\Delta_{n-1}\times\Delta_{n-1}\colon \Pi(x,y)=1\}.
\ee 
Define
\be
P(\Pi(X,Y)\leq t)=\mu\otimes\mu\left(\left\{(x,y)\in\Delta_{n-1}\times\Delta_{n-1}\colon \Pi(x,y)\leq t\right\}\right).
\ee
The next result is a restatement of Eq.~\eqref{eq:main2_maj} of Theorem~\ref{thm:intro_maj}.
\begin{thm}
\label{thm:limit} Let $\mu_n$ be the uniform probability measure on the simplex $\Delta_{n-1}$. Then, 
\be
\label{eq:conv_1}
\lim_{n\to\infty}\Pi(X,Y)=1,
\ee
in probability.
\end{thm}
The proof is based on formula~\eqref{eq:propVidalMaj}. The idea is that if $X$ is uniform in the simplex, then for each fixed $\epsilon>0$, $\phi_k(X)$ is essentially equal to $(1 \pm \epsilon)$ times its mean, aside from an event with probability that is small even when we sum over $k=1,\ldots,n$.  This then shows that $\phi_k(X)/ \phi_k(Y) = (1\pm \epsilon)$ for all $k$ with probability tending to $1$.  Since this holds for each $\epsilon > 0$, it shows $\Pi(X,Y)$ converges to $1$ in probability.
The rate of convergence (i.e. how small $\epsilon$ can be taken it we allow it to decrease to $0$ as $n\to\infty$) seems quite slow; from the proof is appears that $\epsilon$ must be at least $1/\mathrm{poly}(\log n)$. It is therefore tricky to see evidence of the convergence~\eqref{eq:conv_1} in a simulation without taking $n$ to be fairly large.  See \figurename~\ref{fig:LimitDistributions}).
\begin{rem}
A  function similar to $\Pi(\cdot,\cdot)$ can be considered for the converse relation $\succ$, defined analogously to \eqref{eq:Pi}, but in terms of the sums of the smallest components of $x$ and $y$. Such a function appears in the resource theory of entanglement and its statistical behaviour was studied in~\cite{cunden2020,Jain24}.
\end{rem}
\begin{figure}
        \includegraphics[width=.49\textwidth]{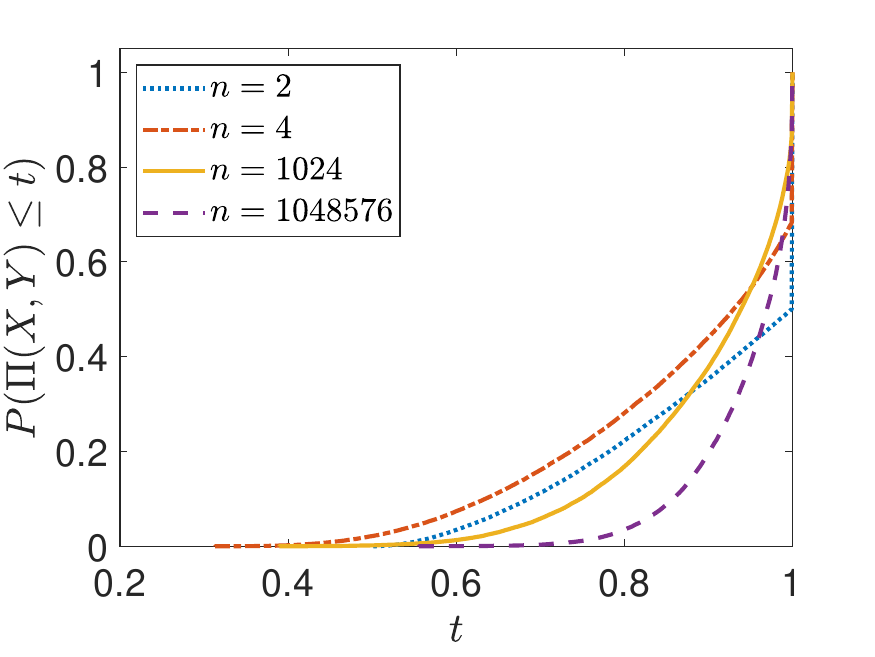}~
        \includegraphics[width=.49\textwidth]{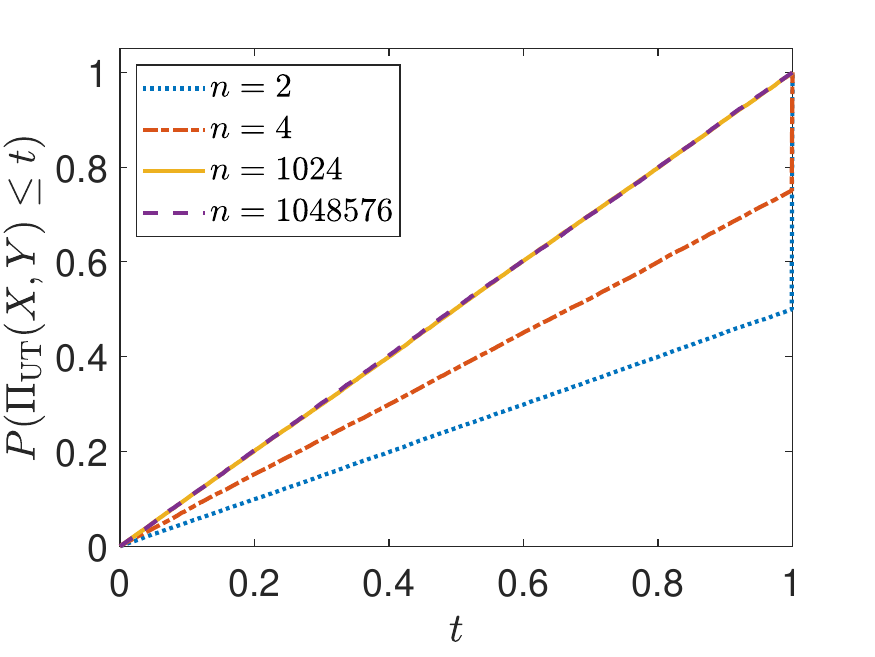}
        \captionsetup{width=1\linewidth}  
	\caption{Empirical cumulative distribution of (a) $\Pi(X,Y)$ and (b) $\Pi_{\text{UT}}(X,Y)$, obtained from ${10^5}$ samples of uniformly random pairs $(x_i, y_i)$ in the simplex $\Delta_{n-1}$.{One can see that the distribution of $\Pi(X,Y)$ concentrates at $1$ as $n\to\infty$, at a slow rate.}}
    \label{fig:LimitDistributions}
\end{figure}

\textit{Mutatis mutandis}, the same question can be posed for UT-majorization. 
Let 
\be
\Pi_{\UT}:\Delta_{n-1}\times \Delta_{n-1}\to[0,1],
\ee
defined as
\be
\Pi_{\UT}(x,y):=\max\left\{p\in[0,1]\colon px+(1-p)z\prec^{\UT} y,\quad\text{for some $z\in\Delta_{n-1}$}\right\}.
\ee 
We have the following simple representation whose proof mimicks the one of Proposition~\ref{prop:VidalMaj}.
\begin{prop}
\label{prop:VidalUT}
\be
\Pi_{\UT}(x,y)=\min_{1\leq k\leq n}\frac{\phi_k^{\UT}(y)}{\phi_k^{\UT}(x)}.
	\ee
where $\phi_k^{\UT}$ are defined in~\eqref{eq:frameUT}.\
\end{prop}
In this case, $0\leq\Pi_{\UT}(x,y)\leq1$. The minimum $\Pi_{\UT}(x,y)=0$ is achieved for instance when $y_1=0$ but $x_1>0$, while $\Pi_{\UT}(x,y)=1$ if and only if $x\prec^{\UT} y$, so that
\be
\Gamma(\Delta_{n-1},\prec^{\UT})=\{(x,y)\in\Delta_{n-1}\times\Delta_{n-1}\colon \Pi_{\UT}(x,y)=1\}.
\ee 
Define
\be
P\left(\Pi_{\UT}(X,Y)\leq t\right)=\mu\otimes\mu\left(\left\{(x,y)\in\Delta_{n-1}\times\Delta_{n-1}\colon \Pi_{\UT}(x,y)\leq t\right\}\right).
\ee
When $X,Y$ are uniformly distributed, $\Pi_{\UT}(X,Y)$ has a (surprisingly simple!) explicit distribution.
\begin{thm} \label{thm:cumulativeDist} Let $\mu_n$ be the uniform measure on $\Delta_{n-1}$. Then, for all $n\geq2$, 
\be\label{eq:MainResult}
P\left(\Pi_{\UT}(X,Y)\leq t\right)=
\begin{cases}
0&\text{if $t\leq0$}\\
\left(1-\dfrac{1}{n}\right)t&\text{if $0<t<1$}\\
1&\text{if $t\geq1$}\\
\end{cases}
\ee
for all $t\in\R$. In particular, as $n\to\infty$, $\Pi_{\UT}(X,Y)$ converges in distribution to a uniform random variable in $[0,1]$.
\end{thm}
The Theorem above corresponds to  Eq.~\eqref{eq:main2} of Theorem~\ref{thm:intro}. 
When $\mu_n$ is the uniform measure, Theorem~\ref{thm:VanishingProb} is  a Corollary of Theorem~\ref{thm:cumulativeDist}. 
\begin{rem}
Equation~\eqref{eq:MainResult} can be expressed in purely geometrical terms. The volume of the $2(n-1)$-dimensional polytope
\be
\left\{(x,y)\in\Delta_{n-1}\times\Delta_{n-1}\colon t (x_1+\cdots+x_k)\leq y_1+\cdots+y_k\,\text{for all $1\leq k \leq n$}\right\}
\ee
is the right-hand side of~\eqref{eq:MainResult}  times the volume of $\Delta_{n-1}\times\Delta_{n-1}$ (that is $(n-1)!\times (n-1)!$). A related polytope was studied in detail by Stanley and Pitman in~\cite{Stanley02}, and their paper was a great source of information in our investigations.
\end{rem}
 
Note that~\eqref{eq:MainResult} is \emph{not} a universal (distribution independent) result. 
We have for instance considered the family of permutation invariant Dirichlet distribution {$\mu_n=\operatorname{Dir}(\alpha,\ldots,\alpha)$}, with density on the unit simplex given by
\be
f_{X_1,\dots,X_{n}}(x)=\frac{\Gamma(n\alpha)}{\Gamma(\alpha)^n} \prod_{i=1}^n x_i^{\alpha-1}\bm{1}_{\Delta_{n-1}}(x).
\ee

The above family includes the uniform distribution on $\Delta_{n-1}$ for $\alpha=1$. One can verify by explicit computation for small $n$, that the distribution of $\Pi_{\UT}(X,Y)$ when $\alpha\neq1$ is not given by~\eqref{eq:MainResult}. For example, when $n=3$ and $\alpha=2$,
$$
P\left(\Pi_{\UT}(X,Y)\leq t\right)=
\frac{10}{7} t^2 - \frac{10}{7} t^3 + \frac{10}{9} t^4 - \frac{4}{9} t^5,\quad\text{ $0<t<1$}.
$$
We did not manage to unveil a closed form for general $\alpha$ and $n$. However, using the known `stick-breaking' representation of Dirichlet distributed random variables {(see Appendix~\ref{app:stick})  one can get the following simple bound.}
\begin{prop} 
\label{prop:Dirichlet}
Let {$\mu_n=\operatorname{Dir}(\alpha,\ldots,\alpha)$} on $\Delta_{n-1}$. Then, for all $n\geq2$,
\be
P\left(\Pi_{\UT}(X,Y)\leq t\right)\leq \left(1-\frac{1}{\alpha n}\right)t.
\ee
\end{prop}

\section{Proofs}
In this section we will use some standard notation and terminology. For two random variables $X$ and $Y$, we write $X\stackrel{d}{=}Y$ to denote \textbf{identity in distribution}. We say that  $X_1,\ldots,X_n$ are \textbf{independent and identically distributed (i.i.d.)} random variables, if they are jointly independent and $X_i\stackrel{d}{=}X_j$ for all $i,j$. We say that $X_1,\ldots,X_n$ are \textbf{exchangeable}, if their joint law is invariant under permutations, $(X_1,\ldots,X_n)\stackrel{d}{=}(X_{\sigma(1)},\ldots,X_{\sigma(n)})$, for all $\sigma\in S_n$. Finally, $Z$ is an \textbf{exponential} random variable \textbf{with rate $1$}, if it has density $f_Z(z)=e^{-z}$ for $z\geq0$, and zero otherwise.
\subsection{Proof of Theorem~\ref{thm:VanishingProb}}
 Pick $X,Y$ independently according to $\mu$. Since UT-majorization is equivalent to unordered majorization, we want to compute
\be
\label{eq:majUTwalk}
P\left(X\leq^{\UT} Y\right)=P\left(\sum_{i=1}^k(Y_i-X_i)\geq0, \text{for all $k=1,\ldots n$}\right).
\ee
Let $Z:=Y-X$ and define
\be
S_0=0, \quad S_{k}=S_{k-1}+Z_{k}, \quad \text{for $k=1,\ldots, n$}.
\ee
Note that  $(S_k)_{0\leq k\leq n}$ is a random bridge starting and ending at $0$, i.e., $S_0=S_n=0$. 
Hence, the probability~\eqref{eq:majUTwalk} can be seen as the \emph{persistence probability} above level $0$ of the random bridge $(S_k)_k$:
\be
P\left(X\leq^{\UT} Y\right)=P\left(S_k\geq0, \text{for all $k=1,\ldots,n$}\right).
\ee
The steps $Z_1,\ldots, Z_n$ are exchangeable and we have $P\left(S_k=0\right)=0$ for all $k=1,\ldots n-1$, since $\mu$ is absolutely continuous. Hence, by Sparre-Andersen theorem~\cite[Theorem 2]{Sparre-Andersen53}, denoting by $N_n$ the number of sums $S_1,\ldots, S_n$ which are $>0$, we have
 $$
 P\left(N_n=m\right)=\frac{1}{n},
 $$
 for all $m=1,\ldots,n$. In particular $ P\left(N_n=n\right)=P\left(S_k\geq0, \text{for all $k=1,\ldots,n$}\right)=1/n$.

\subsection{Preparatory propositions}

The starting point
is that if $X$ is uniformly distributed on the simplex, then its components and the decreasing rearrangement
of its components have a closed form description in terms of independent exponential random variables (see, for instance,~\cite{Devroye1986} or~\cite[Theorems 3.1 and 3.2]{Jain24}). These classical results  provide  convenient representations of the monotones  $\phi_k(X)$ and $\phi_k^{\UT}(X)$.
\begin{prop}
\label{prop:repr_phik}
	Let $X=(X_1,\dots,X_n)$ be a uniform point on the simplex $\Delta_{n-1}$, and $Z_1,\ldots,Z_n$ be i.i.d. exponential random variables with rate $1$.  
	\begin{enumerate}
	\item 
	Denote $\phi_k(X):=\sum_{i=1}^kX^{\downarrow}_i$ for $k=1,\ldots,n$. Then,
	\be
	\label{eq:V_1}
	\left(\phi_1(X),\phi_2(X),\ldots,\phi_n(X)\right)\stackrel{d}{=}\frac{1}{V_n}\left(V_1,V_2,\ldots,V_n\right),
	\ee
	where $V_k=a_{k,1}Z_1+a_{k,2}Z_2+\cdots+a_{k,n}Z_n$, with
	\be
	\label{eq:V_2}
	a_{k,i}=
	\begin{cases}
	1 &\text{for $i\leq k$}\\
	\frac{k}{i}&\text{for $i> k$}
	\end{cases}.
	\ee
	\item  Denote $\phi_k^{\UT}(X):=\sum_{i=1}^kX_i$ for $k=1,\ldots,n$. Then,
	\be
	\left(\phi_1^{\UT}(X),\phi_2^{\UT}(X),\ldots,\phi_n^{\UT}(X)\right)\stackrel{d}{=}\frac{1}{W_n}\left(W_1,W_2,\ldots,W_n\right),
	\ee
	where $W_k=b_{k,1}Z_1+b_{k,2}Z_2+\cdots+b_{k,n}Z_n$, with
	\be
	b_{k,i}=
	\begin{cases}
	1 &\text{for $i\leq k$}\\
	0&\text{for $i> k$}
	\end{cases}.
	\ee
	\end{enumerate}
\end{prop}
A more convenient representation for the monotones $\phi_k^{\UT}(X)$ is provided by a connection  between partial sums of a uniform point on the simplex and the order statistics of independent uniform random variables. If $U_1,\dots,U_{n-1}$ are ${n-1}$ independent random variables, we denote the corresponding order statistics with $U_{({n-1},1)},\ldots U_{({n-1},{n-1})}$, where $U_{({n-1},i)}$ is the $i-$th smallest among the $U_j$'s, that is the ${n-1}$ random variables are rearranged in nondecreasing order: $U_{({n-1},1)}\leq \cdots\leq U_{({n-1},{n-1})}$. 

Then, using classical representation theorems for order statistics (see, e.g.,~\cite{Hajos54,ArnoldOrder}) we have the following.
\begin{prop}
\label{prop:alternative}
	Let $X=(X_1,\dots,X_n)$ be a uniform point on the simplex $\Delta_{n-1}$, and $U_1,\ldots,U_{n-1}$ be i.i.d. uniform random variables on the interval $[0,1]$.  Denote $\phi_k^{\UT}(X):=\sum_{i=1}^kX_i$ for $k=1,\ldots,n$. Then,
		\be
		\label{eq:order_repr}
	\left(\phi_1^{\UT}(X),\phi_2^{\UT}(X),\ldots,\phi_n^{\UT}(X)\right)\stackrel{d}{=}\left(U_{({n-1},1)},\dots,U_{({n-1},{n-1})},1\right).
	\ee
 Since $X_j=\phi_j^{\UT}(X)-\phi_{j-1}^{\UT}(X)$ (with $\phi_{0}^{\UT}(X):=0$), this is also equivalent to say that $X$ has the same distribution of the sequence of spacings between order statistics of ${n-1}$  uniform random variables.

\end{prop}

\begin{rem}
	If $U_1,\dots,U_{n-1}$ are i.i.d. uniform random variables on $[0,1]$,
	$(U_{(n-1,1)},\dots,U_{(n-1,n-1)})$ is a uniform point over the simplex $\mathcal{E}_{n-1}:=\{u\in\mathbb{R}^{n-1}: 0\leq u_1\leq u_2 \leq \cdots \leq u_{n-1}\leq 1\}$.
\end{rem}

\subsection{Proof of Theorem~\ref{thm:limit} }
The technical ingredient to prove Theorem~\ref{thm:limit} is the following bound.
\begin{prop}[Bernstein's inequality {(see, for example, ~\cite[Theorem 2.8.2]{vershynin2018high})}] Let $Z_1,\ldots,Z_n$ be i.i.d. exponential random variables with rate $1$, and let $V= a_1Z_1+\cdots +a_nZ_n$ for some coefficients $a_1,\ldots,a_n\in\R$. There is an absolute constant $c>0$ such that
\be
P\left(|V-E V|\geq t\right)\leq \exp\left(-c\min\left(\frac{t^2}{a_1^2+\cdots+a_n^2},\frac{t}{\max|a_i|}\right)\right),
\ee
for every $t\geq0$.
\end{prop}
\par
\begin{cor} 
\label{cor:applied_Bernstein} Let $V_k$ defined as in ~\eqref{eq:V_1}-\eqref{eq:V_2}. Then,
for all $k\leq n$,
\be
P\left(|V_k-E V_k|\geq \epsilon E V_k\right)\leq \exp\left(-c\min\left(\epsilon^2{A_k},\epsilon  B_k\right)\right),
\ee
where 
\be
A_k:=\frac{\left(\sum_{i=1}^n a_{k,i}\right)^2}{\sum_{i=1}^n a_{k,i}^2}, \quad 
B_{k}:= \frac{\sum_{i=1}^n a_{k,i}}{\max_{1\leq i\leq n} a_{k,i}}.
\ee
\end{cor}
We denote by $H_n=\sum_{i=1}^n\frac{1}{i}$ the $n$-th harmonic number, and $H_{n,m}=\sum_{i=1}^n\frac{1}{i^m}$ the $n$-th harmonic number of order $m$.
\begin{lem}
For all $k=1,\ldots, n$,
\begin{align}
\sum_{i=1}^n a_{k,i}=k\left(1+H_n-H_k\right),\quad 
\sum_{i=1}^n a_{k,i}^2=k\left(1+k(H_{n,2}-H_{k,2})\right),\quad 
\max_{1\leq i\leq n} a_{k,i}=1.
\end{align}
\end{lem}
\begin{lem} \label{lem:bounds} 
We have
\begin{align}
\label{eq:bound_Ak}
A_k&\geq  \frac{1}{2} \left(\log\frac{n}{2}\right)^2+\alpha_n(k-1),\\
\label{eq:bound_Bk}
B_k &\geq \log \frac{n}{2}+\beta_n(k-1)
 \end{align}
 where
\begin{align}
\alpha_{n}= \frac{n}{n-1}\left(1+ \frac{ \left(1+\log \left(\frac{n+1}{2}\right)\right)^2}{2 n-1}\right),\qquad 
\beta_{n}=1-\frac{\log \left(\frac{n+1}{2}\right)}{n-1}.
\end{align}
\begin{proof}
We start from the following bounds,
\begin{align}
\label{eq:H1}
H_n-H_k= \sum_{i=k+1}^n\frac{1}{i}&\geq \int_{k+1}^{n+1}\frac{dx}{x}= \log \frac{n+1}{k+1}, \\
\label{eq:H2}
H_{n,2}-H_{k,2}=  \sum_{i=k+1}^n\frac{1}{i^2}&\leq  \int_{k}^{n}\frac{dx}{x^2} =\frac{n-k}{nk}.
\end{align}
Therefore,
\begin{align}
A_k=\frac{k\left(1+H_n-H_k\right)^2}{\left(1+k(H_{n,2}-H_{k,2})\right)}\geq  \frac{nk}{ (2 n-k)}\left(1+\log  \frac{n+1}{k+1}\right)^2.
\end{align}
The function $x\mapsto g(x):=\frac{nx}{ (2 n-x)}\left(1+\log  \frac{n+1}{x+1}\right)^2$ is concave for $x\in[1,n]$. Hence, it can be bounded from below by the linear interpolation $g(1)+\frac{g(n)-g(1)}{n-1}(x-1)$. This gives
\begin{align}
A_k\geq \frac{n \left(1+\log \left(\frac{n+1}{2}\right)\right)^2}{2 n-1}+ \frac{ \left(n+ \frac{n \left(1+\log \left(\frac{n+1}{2}\right)\right)^2}{2 n-1}\right)}{n-1}(k-1),
\end{align}
and hence~\eqref{eq:bound_Ak}. The proof of ~\eqref{eq:bound_Bk} for $B_k$ is similar. From~\eqref{eq:H1},
\begin{align}
B_k=k\left(1+H_n-H_k\right)\geq k \left(1 + \log \frac{n+1}{k+1}\right).
\end{align}
The map $x\mapsto h(x):=x \left(1 + \log \frac{n+1}{x+1}\right)$ is concave for $x\in[1,n]$. We can therefore write
\begin{align}
B_k\geq \log \left(\frac{n+1}{2}\right)+1+\frac{\left(n-\log \left(\frac{n+1}{2}\right)-1\right)}{n-1}(k-1),
\end{align}
and hence~\eqref{eq:bound_Bk}. 
\end{proof}
\end{lem}
\begin{proof}[Proof of Theorem~\ref{thm:limit}]
We want to show that
$
P\left(\Pi(X,Y)<1-\epsilon\right)\to 0,
$ {as $n\to\infty$}.
Let $\epsilon\in(0,1)$.
We have
\begin{multline*}
P\left(\Pi(X,Y)<1-\epsilon\right)=P\left(\min_{1\leq k\leq n}\frac{\phi_k(Y)}{\phi_k(X)}<1-\epsilon\right)\\
=P\left(\frac{\phi_k(Y)}{\phi_k(X)}<1-\epsilon,\,\,\text{for some $k\leq n$}\right)
\leq \sum_{k=1}^nP\left(\frac{\phi_k(Y)}{\phi_k(X)}<1-\epsilon\right),
\end{multline*}
by the union bound.
Using the representation of Proposition~\ref{prop:repr_phik}, we want to bound
\be
\label{eq:want_bound}
P\left(\frac{\phi_k(Y)}{\phi_k(X)}<1-\epsilon\right)=P\left(\frac{V_k(Y)}{V_n(Y)}<(1-\epsilon)\frac{V_k(X)}{V_n(X)}\right)\leq\delta_{n,k},
\ee
so that, 
\be
\label{eq:want_sum}
\sum_{k=1}^n\delta_{n,k}\to0,\quad\text{as $n\to\infty$}.
\ee
By Corollary~\ref{cor:applied_Bernstein}, 
$
\left|V_n-n\right|\leq n\epsilon$ with probability at least $1-e^{-c\epsilon^2 n}
$.
Hence,
\begin{align*}
&P\left(\frac{V_k(Y)}{V_n(Y)}<(1-\epsilon)\frac{V_k(X)}{V_n(X)}\right)\\
&\leq P\left(V_k(Y)<(1-\epsilon)\frac{V_n(Y)}{V_n(X)}V_k(X),\quad (V_n(X),V_n(Y))\in \left((1-\epsilon)n,(1+\epsilon)n\right)^2\right)+2e^{-c\epsilon^2 n}\\
&\leq P\left(V_k(Y)<(1+\epsilon)V_k(X)\right)+2e^{-c\epsilon^2 n}\\
&\leq P\left(V_k(Y)<(1+\epsilon)V_k(X),\quad V_k(X)\in\left((1-\epsilon)EV_k(X),(1+\epsilon)EV_k(X)\right)\right)\\
&+P\left(V_k(Y)<(1+\epsilon)V_k(X),\quad V_k(X)\notin\left((1-\epsilon)EV_k(X),(1+\epsilon)EV_k(X)\right)\right)+2e^{-c\epsilon^2 n}\\
&\leq P\left(V_k(Y)<(1+\epsilon)^2EV_k(X)\right)
+e^{-c\min\left(\epsilon^2 A_k,\epsilon B_k\right)}+2e^{-c\epsilon^2 n}\\ 
&\leq P\left(|V_k(Y)-EV_k(Y)|<3\epsilon EV_k(Y)\right)
+e^{-c\min\left(\epsilon^2 A_k,\epsilon B_k\right)}+2e^{-c\epsilon^2 n}\\
&{\leq 4e^{-c\min\left(\epsilon^2 A_k,\epsilon B_k\right)},}
\end{align*}
where in the last line we used $A_k\leqslant n$. 
\par
Therefore we have the bound~\eqref{eq:want_bound} with 
\be
\delta_{n,k}:=4e^{-c\epsilon^2 A_k}+4e^{-c \epsilon B_k}.
\ee
It remains to show~\eqref{eq:want_sum}. 
For this we use Lemma~\ref{lem:bounds}:
\begin{align*}
\sum_{k=1}^ne^{-c\epsilon^2 A_{k}}&\leq e^{-c(\epsilon\log (n/2))^2/2}\sum_{k=0}^{n-1}e^{-c\epsilon^2  \alpha_n k}
=e^{-c(\epsilon\log (n/2))^2/2}\frac{1-e^{-c\epsilon^2 n\alpha_n}}{1-e^{-c\epsilon^2  \alpha_n}}\to 0,\\
\sum_{k=1}^ne^{-c\epsilon B_{k}}&\leq e^{-c\epsilon \log (n/2)}\sum_{k=0}^{n-1}e^{-c\epsilon  \beta_nk}
= e^{-c\epsilon \log (n/2)}\frac{1-e^{-c\epsilon  n\beta_n}}{1-e^{-c\epsilon  \beta_n}}\to 0,
\end{align*}
as $n\to\infty$.
\end{proof}

\subsection{Proof of Theorem~\ref{thm:cumulativeDist} and Proposition~\ref{prop:Dirichlet}  }
The connection with order statistics given in Proposition~\ref{prop:alternative} is useful for our purposes since a large body of literature has been devoted to the study of order statistics of independent random variables (see, e.g., ~\cite{Denuit,OrderStatistics,ArnoldOrder,ShorakWellner}).  
Of special interest for us are the probabilities
\begin{equation}\label{eq:P(a,b)}
	P\left(a_j\leq U_{(m,j)}\leq b_j \text{ for all } 1\leq j \leq m\right),
\end{equation}
where $0\leq a_1\leq \dots \leq a_m\leq 1$ and $0\leq b_1\leq \dots \leq b_m\leq 1$.
\par

From now on, we specialize \eqref{eq:P(a,b)} to the case $b_j=1$ for all $1\leq j\leq m$, 
and we use the notation
\begin{equation}\label{eq:P(a)}
	P_m(a_1,\dots,a_m):=P(U_{(m,j)}>a_j \text{ for all }1\leq j \leq m).
\end{equation}
The key ingredient to prove Theorem~\ref{thm:cumulativeDist} is a recursive relation for \eqref{eq:P(a)} discovered by L. N. Bolshev shortly before his death and later published in~\cite{Kostelnikova}.
\begin{prop}[Bolshev's recursion, \cite{ShorakWellner} \S 9.3]
	Let $0\leq a_1\leq \dots \leq a_m \leq  1$. Then:
	\begin{equation}\label{eq:Bolshev}
		P_m(a_1,\dots,a_m)=1-\sum_{k=1}^m \binom{m}{k} ~a_k^k P_{{m}-k}(a_{k+1},\dots, a_m),
	\end{equation}
	with $P_0\equiv 1$.
\end{prop}
	A proof of this proposition can be found in \cite{Kostelnikova} and \cite{ShorakWellner} in terms of empirical distribution functions. 
	For the reader's convenience, we include a self-contained proof in Appendix~\ref{app:Bolshev}.
\label{sec:proofs}

\begin{proof}[{Proof of Theorem~\ref{thm:cumulativeDist}}]
Let  $F_n(t):=P(\Pi_{\UT}(X,Y)\leq t)$ be the distribution function of $\Pi_{\UT}(X,Y)$. By the normalization  of $X,Y\in\Delta_{n}$,
\begin{equation}
	0\leq \Pi_{\UT}(X,Y)\leq \frac{Y_1+\dots+Y_n}{X_1+\dots+X_n}=1,
\end{equation} 
almost surely, which implies $F_n(t)=0$ for $t<0$, and $F_n(t)=1$ for $t\geq1$. Let $0\leq t <1$, and denote with $S_k:=X_1+\dots+X_{k}$ and $T_k:=Y_1+\dots+Y_{k}$ for $k=1,\dots,n-1$. Then,
\begin{equation}\label{eq:fn}
	F_n(t)=1-P(\Pi_{\UT}(X,Y)>t)=1-P\left(\frac{T_1}{S_1}>t, \dots, \frac{T_{n-1}}{S_{n-1}}>t\right), 
\end{equation}
where we dropped the inequality $T_n/S_n>t$, since it is redundant for $t<1$.
Using Proposition~\ref{prop:alternative} and the notation \eqref{eq:P(a)} we can write
\begin{align}
	1-F_n(t)&=P\left( T_1>t S_1, \dots, T_{n-1}>tS_{n-1}\right)
	\label{eq:gn}
	= E_S [P_{n-1}(tS_1, \dots, t S_{n-1})],
\end{align}
where $E_{S}$ denotes the average with respect to $S_1,S_2,\ldots, S_{n-1}$.
Inserting Bolshev's recursion \eqref{eq:Bolshev} into \eqref{eq:gn}, one has
\begin{equation}\label{eq:gnRec}
	F_n(t)=\sum_{k=1}^{n-1} \binom{{n-1}}{k} t^k E_S \left[S_{k}^k P_{{n-1}-k}(tS_{k+1},\dots, tS_{n-1})\right]. 
\end{equation}
Isolating the first term in~\eqref{eq:gnRec} and relabeling the summation index on the remaining terms one obtains
\begin{multline}
	\label{eq:FnBis}
	F_n(t)
	= (n-1) t E_S\left[ S_1 P_{n-2}(tS_2,\dots, tS_{n-1})\right]\\
	 +\sum_{k=1}^{n-2}\binom{n-1}{k+1} t^{k+1}E_S\left[ S_{k+1}^{k+1} P_{n-k}(tS_{k+2},\dots, tS_{n-1})\right], 
\end{multline}
We now use Bolshev's recursion again on the first term of \eqref{eq:FnBis}, obtaining:
\begin{multline}
	\label{eq:Fnbisbis}
	F_n(t)=(n-1)t E_S S_1-(n-1)t\sum_{k=1}^{n-2}\binom{n-2}{k} t^k E_S \left[S_1 S_{k+1}^{k} P_{n-2-k}(tS_{k+2},\dots, tS_{n-1})\right]\\
	  + \sum_{k=1}^{n-2}\binom{n-1}{k+1} t^{k+1}E_S\left[ S_{k+1}^{k+1} P_{n-k-2}(tS_{k+2},\dots, tS_{n-1})\right]\\
	= (n-1)t E_S S_1+\sum_{k=1}^{n-2} \binom{n-1}{k+1} t^{k+1}  E_S \left[(S_{k+1}-(k+1)S_1 )S_{k+1}^{k} P_{n-k-2}(tS_{k+2},\dots, tS_{n-1})\right].
\end{multline}
The average in the $k$-th term of the summation is given by
\begin{multline}\notag
	E_S \left[(S_{k+1}-(k+1)S_1) S_{k+1}^{k} P_{n-k-2}(tS_{k+2},\dots,tS_{n-1})\right]\\
	 =n! \int_{0}^{1}\de s_{n-1} \int_{0}^{s_{n-1}}\de s_{n-2} \dots \int_{0}^{s_{2}}\de s_1 [s_{k+1}-(k+1)s_1] s_{k+1}^k P_{n-k-2}(t s_{k+2},\ldots, t s_{n-1}),
\end{multline}
where the integration over $s_1,\dots s_{k+1}$ can be carried out explicitly, yielding
\begin{multline*}
	\int_{0}^{s_{k+2}}\de s_{k+1} \dots \int_{0}^{s_{2}}\de s_1 [s_{k+1}-(k+1)s_1] s_{k+1}^k=\int_0^{s_{k+2}}\de s_{k+1} \left(s_{k+1}\frac{s_{k+1}^{k}}{k!}-(k+1)\frac{s_{k+1}^{k+1}}{(k+1)!}\right)=0.
\end{multline*}
Then, all terms of the summation in \eqref{eq:Fnbisbis} vanish. 
Finally, noting that $E_S S_1=(n-1)!/n!=1/n$, we conclude the proof 
$$
	F_n(t)=(n-1)t E_S S_1=\frac{n-1}{n}t=\left(1-\frac{1}{n}\right)t.
$$
\end{proof}

\begin{proof}[Proof of Proposition~\ref{prop:Dirichlet}] Let $\alpha$ be a positive integer. 
As a special case of the stick-breaking representation in Proposition~\ref{prop:stick}, we have that if $U_{1}\ldots,U_{n\alpha -1}$ are independent random variables uniformly distributed over the interval $[0,1]$, then
the vector $(X_1,\ldots,X_n)$ with components
\begin{align*}
    X_1&=U_{(n\alpha -1,\alpha)}\\
    X_2&=U_{(n\alpha -1,2\alpha)}-U_{(n\alpha -1,\alpha)}\\
    &\vdots\\
    X_{n-1}&=U_{(n\alpha -1,(n-1)\alpha)}-U_{(n\alpha -1,(n-2)\alpha)}\\
X_{n}&=1-U_{(n\alpha -1,(n-1)\alpha)},
\end{align*}
is distributed according to $\operatorname{Dir}(\alpha,\ldots,\alpha)$.
Therefore, using the same notation as in the proof of Theorem~\ref{thm:cumulativeDist},
\begin{multline*}
1-F_n(t)=P\left( T_1>t S_1,T_2>t S_2, \dots, T_{n-1}>tS_{n-1}\right)\\
=P\left(V_{(n\alpha-1,\alpha)}\geq t U_{(n\alpha-1,\alpha)},V_{(n\alpha-1,2\alpha)}\geq t U_{(n\alpha-1,2\alpha)},\ldots,V_{(n\alpha-1,(n-1)\alpha)}\geq t U_{(n\alpha-1,(n-1)\alpha)} \right),
\end{multline*}
where $U_1,\ldots,U_{n\alpha-1},V_1,\ldots,V_{n\alpha-1}$ are independent uniform random variables in the interval $[0,1]$. We can easily upper bound the latter probability
\begin{multline}
  P\left(V_{(n\alpha-1,j)}\geq t U_{(n\alpha-1,j)},\text{for all $j\in\{\alpha,2\alpha,\ldots,(n-1)\alpha\}$}\right)  \\
  \geq 
    P\left(V_{(n\alpha-1,j)}\geq t U_{(n\alpha-1,j)},\text{for all $j\in\{1,2,\ldots,n\alpha-1\}$}\right)
    =1-\frac{n\alpha -1}{n\alpha}t,
\end{multline}
    where the last equality follows from Theorem~\ref{thm:cumulativeDist} applied to pairs of uniform random points in   $\Delta_{\alpha n-1}$.
    \end{proof}

\appendix

\section{Proof of Bolshev's recursion}
\label{app:Bolshev}

We denote by 
	\begin{equation}\label{eq:EDFUnif}
		G_m(t):=\frac{1}{m}\sum_{k=1}^m \bm{1}_{(-\infty,t]}(U_k),
	\end{equation}
	the empirical distribution function of $m$  independent random variables $U_1,\dots,U_m$ uniformly distributed in $[0,1]$.
\par

The probability of the event $\{U_{(m,1)}>a_1,\dots,U_{(m,m)}>a_m\}$ can be found by partitioning the sample space into the following disjoint union of events:
	\begin{equation}\label{eq:partition}
		\{U_{(m,1)}>a_1,\dots,U_{(m,m)}>a_m\}\cup\bigcup_{k=1}^{n}\{U_{(m,k)}\leq a_k,\ U_{(m,j)}>a_j \text{ for all }j>k\}.
	\end{equation}
	Informally, we first isolate the event $\{U_{(m,1)}\leq a_1\}$, then we partition the complementary event according to the position of $U_{(m,2)}$ into $\{U_{(m,1)}>a_1,U_{(m,2)}\leq a_2\}$ and $\{U_{(m,1)}>a_1,U_{(m,2)}>a_2\}$, and we iterate the procedure up to $m$. From \eqref{eq:partition} we obtain
	\begin{align}\notag
		&P_m(a_1,\dots,a_m)=1-\sum_{k=1}^m P\left(\{U_{(m,k)}\leq a_k\}\cap\{ U_{(m,j)}>a_j  \text{ for all }j>k\}\right)\\
		&\qquad=1-\sum_{k=1}^m P(mG_m(a_k)=k)P(mG_m(a_j)\leq j-1  \text{ for all }j>k|mG_m(a_k)=k),\notag
	\end{align}
	where in the last equality we used the empirical spectral distribution \eqref{eq:EDFUnif} and the fact that  $\{U_{(m,k)}\leq a_k\}\cap \{U_{(m,k+1)}>a_{k+1}\}$ with $a_{k+1}>a_k$ selects the event where exactly $k$ among the $m$ random variables $(U_1,\dots,U_m)$ fall within $[0,a_k]$. The corresponding probability is given by the binomial distribution $\operatorname{Bin}(m,a_k)$, so that we have:
	\begin{equation}
		P_m(a_1,\dots,a_m)=1-\sum_{k=1}^m \binom{m}{k}a_k^k(1-a_k)^{m-k}P(mG_m(a_j)\leq j-1  \text{ for all }j>k|mG_m(a_k)=k)
		\label{eq:almostBolshev}
	\end{equation}
	Now notice that 
	\begin{align}\notag
		&P(mG_m(a_j)\leq j-1  \text{ for all }j>k|mG_m(a_k)=k)\\&\qquad=P((m-k)G_{m-k}(a_j)\leq j-1-k  \text{ for all }j>k|(m-k)G_{m-k}(a_k)=0),
	\end{align}
	that is,  the locations of the largest $(m-k)$ among $(U_1,\dots,U_m)$ given that $k$ of them fall below $a_k$ are identically distributed
	to the locations of the order statistics of $(m-k)$ uniform random variables $(V_1,\dots,V_{m-k})$ given that \textit{none} of them falls below $a_k$.
	Note also that $(1-a_k)^{m-k}$ corresponds to the probability that among $m-k$ uniform random variables in $[0,1]$ none of them falls in $[0,a_k]$. Therefore:
	\begin{align}\notag
	&(1-a_k)^{m-k}P(mG_m(a_j)\leq j-1  \text{ for all }j>k|mG_m(a_k)=k)\\ \notag&\qquad=P((m-k)G_{m-k}(a_j)\leq j-1-k  \text{ for all }j>k)\\&\qquad=P(V_{(m-k,1)}>a_{k+1},\dots, V_{(m-k,n-k)}>a_{m})\notag
	\\&\qquad=P_{m-k}(a_{k+1},\dots,a_m),
	\label{eq:KeyRecursion}
	\end{align}
	Substitution of \eqref{eq:KeyRecursion} into \eqref{eq:almostBolshev} yields the claimed recursive relation.

\section{Dirichlet distribution and stick-breaking representation}
\label{app:stick}
The Dirichlet measure with parameters $\alpha_1,\ldots,\alpha_n>0$, denoted $\operatorname{Dir}(\alpha_1,\ldots,\alpha_n)$, is the probability measure on the unit simplex $\Delta_{n-1}$ with density
 \be
\label{eq:density_Dir_gen}
f_{X_1,\ldots,X_n}(x)=\frac{\Gamma(\alpha_1+\cdots+\alpha_n)}{\Gamma(\alpha_1)\cdots\Gamma(\alpha_n)}\prod_{i=1}^nx_i^{\alpha_i-1}\bm{1}_{\Delta_{n}}(x).
 \ee
 For $n>2$, the Dirichlet distribution is a natural multidimensional generalisation of the Beta distribution. 

 When $\alpha_1,\ldots,\alpha_n$ are positive integers, a generalisation of~Proposition~\ref{prop:alternative}  gives a nice and useful representation of a Dirichlet distributed random vector in terms of spacings of order statistics of independent uniform random variables in the interval $[0,1]$.
 \begin{prop}[Stick-breaking representation of the Dirichlet distribution] 
 \label{prop:stick}
     Let $\alpha_1,\ldots,\alpha_n$ be positive integers.  Denote $m_1=\alpha_1$, $m_2=\alpha_1+\alpha_2$, \ldots, $m_{n-1}=\alpha_1+\cdots+\alpha_{n-1}$, and $m=\alpha_1+\cdots+\alpha_n$.  Let $U_1,\ldots,U_{m-1}$ be independent  random variables that are uniformly distributed over the interval $[0,1]$
 and $U_{(m-1,1)}\leq\cdots\leq U_{(m-1,m-1)} $
  the corresponding order statistics. 
 Then, the vector $\left(X_1,X_2,\ldots,X_n\right)$ of the $n$ spacings
  \begin{align*}
  X_1&=U_{(m-1,m_1)}\\
    X_2&=U_{(m-1,m_2)}-U_{(m-1,m_1)}\\
    &\vdots\\
    X_{n-1}&=U_{(m-1,m_{n-1})}-U_{(m-1,m_{n-2})}\\
    X_{n}&=1-U_{(m-1,m_{n-1})}
  \end{align*}
 is distributed according to $\operatorname{Dir}(\alpha_1,\ldots,\alpha_n)$.
 \end{prop}

\section*{Acknowledgements}
We are grateful to F. vom Ende for bibliographical remarks, to M. Michelen for precious insights leading to Theorem~\ref{thm:limit}, and to an anonymous referee for valuable suggestions that significantly improved the paper.
\section*{Declarations}
\subsection*{Funding and/or Conflicts of interests/Competing interests} The authors have no competing interests to declare that are relevant to the content of this article.

\subsection*{Data availability} 
Data sharing not applicable to this article as no datasets were generated or analysed during the current study.

\bibliographystyle{alpham}
\bibliography{References}

\end{document}